\documentclass[a4paper]{article}
\usepackage{amsmath}

\newtheorem{theorem}{Theorem}

\newtheorem{definition}{Definition}

\newtheorem{example}{Example}
\newenvironment{proof}[1][Proof]{\noindent\textbf{#1.} }{\ \rule{0.5em}{0.5em}}

\begin{document}
\title{\textbf{Quantum codes from codes over Gaussian integers with respect to the Mannheim metric}}
\author{ Mehmet \"{O}zen, Murat G\"{u}zeltepe  \\
{\small Department of Mathematics, Sakarya University, TR54187
Sakarya, Turkey }}
\date{}
\maketitle

\begin{abstract}
In this paper, some nonbinary quantum codes using classical codes
over Gaussian integers are obtained. Also, some of our quantum
codes are better than or comparable with those known before, (for
instance $[[8,2,5]]_{4+i}$).
\end{abstract}




\section{Introduction } An important class of quantum codes are Calderbank-Shor-Steane
(shortly CSS) codes. In fact, CSS codes are obtained from two
classical codes such that one of these codes contains the other
code. Moreover, the bit flip and the phase flip error correcting
capacities of a CSS code depends on the classical code that
contains the other code and the dual code of the other classical
code, respectively [1, pp. 450-451]. The possibility of correcting
decoherence errors in entangled states was discovered by Shor
\cite{2} and Steane \cite{3}. Binary quantum CSS codes have been
constructed in several ways (for instance \cite{3,10,11,12}). In
\cite{10}, good quantum codes of minimum distance three and four
for such length $n$ are obtained via Steane's construction and the
CSS construction. In \cite{11}, a large number of good quantum
codes of minimum distance five and six by Steane's Construction
were given. In \cite{12}, some quantum error correcting codes,
including an optimal quantum code $\left[ {\left[ {27,13,5}
\right]} \right]$ , were presented. Later, some results were
generalized to the case of nonbinary stabilizer codes
\cite{5,6,7,8}. A connection between classical codes and nonbinary
quantum codes was given in \cite{5,6,7}. However, the theory
explained in \cite{5,6,7} is not nearly as complete as in the
binary case. The closest theory to the binary case of nonbinary
stabilizer codes was presented in \cite{8}.\

\

On the other hand, the  Mannheim metric was introduced by Huber in
\cite{9}. It is well known that the Euclidean metric is the
relevant metric for maximum-likelihood decoding. Although the
Mannheim metric is a reasonable approximation to it, it is not a
priori, a natural choice. However, the codes being proposed are
very useful in coded modulation schemes based on quadrature
amplitude modulation (QAM)-type constellations for which neither
Hamming nor Lee metric is appropriate.\
\

The rest of this paper is organized as follows. In section 2,
classical codes over Gaussian integer ring with respect to the
Mannheim metric are given. In Section 3, error bases are defined
and quantum codes with respect to the Mannheim distance are
constructed.

\section{Codes over Gaussian integers}

 Gaussian integers are a subset of complex numbers which have integers as real and imaginary
 parts. Let $p = {a^2} + {b^2} = \pi \overline \pi =N(\pi)  \equiv 1\;\left( {\bmod 4}
 \right)$, where $\pi  = a + ib$ is a Gaussian integer, $\overline \pi   = a -
 ib$ denotes the conjugate of $\pi$ and $p$ is an odd prime integer. Here, $N(\pi)$ denotes the conjugate of $\pi$. Let $G$ denotes the Gaussian integers and ${G_\pi }$ the residue class of
 $G$ modulo $\pi$, where the modulo function $\mu :G \to {G_\pi }$
 is defined according to \begin{equation}\label{eq:1}\mu (\varsigma ) = \varsigma \;\bmod \;\pi  = \varsigma  - \left[ {\frac{{\varsigma \overline \pi  }}{{\pi \overline \pi  }}}
 \right]\pi.\end{equation} $\left[  \cdot  \right]$ denotes rounding of complex numbers. The rounding
 of complex numbers to Gaussian integers can be done by rounding the real and imaginary parts separately to the closest
 integer. Hence, $G_\pi$ becomes a finite field with characteristic
 $p$.
    Let $\alpha ,\beta  \in {G_\pi }$ and $\gamma  = \beta  - \alpha \;(\bmod \pi
    )$. Then, the Mannheim weight of $\gamma$ is defined as ${w_M}(\gamma ) = \left| {{\mathop{\rm Re}\nolimits} \left( \gamma  \right)} \right| + \left| {{\mathop{\rm Im}\nolimits} \left( \gamma  \right)}
    \right|$. Also, the Mannheim distance $d_m$ between $\alpha$
    and $\beta$ is defined as $d_M(\alpha, \beta)=w_M(\gamma)$.
    Let $C$ be code of length $n$ over $G_\pi$ and let $c = \left( {\begin{array}{*{20}{c}}
   {{c_0,}} & {{c_1,}} &  \cdots,  & {{c_{n - 1}}}  \\
\end{array}} \right)$ be a codeword. Then the Mannheim weight of $c$ is equal to $\sum\limits_{i = 0}^{n - 1} {\left( {\left| {{\mathop{\rm Re}\nolimits} ({c_i})} \right| + \left| {{\mathop{\rm Im}\nolimits} ({c_i})} \right|} \right)}. $
 Note that A Mannheim error of weight one takes on one of the
four values $ \pm 1,\; \pm i$ \cite{9}. It is well known that the
Hamming weight of $c$ is the number of
    the non-zero entries of $c$. We give an example to compare a
    classical code with respect to these metrics.

 \begin{example}Let $p=17$. Then, $$G_{4+i}=\left\{ {0,\pm 1,\pm i,\pm 2, \pm 2i \pm (1+i),\pm (1-i),\pm (2-i),\pm (1+2i)} \right\}.$$ Let the generator matrix of $C$ over $G_{4+i}$ be $\left(
\begin{array}{cc}
  -1+i, & 1 \\
\end{array}
\right)$. Then, the set of the codewords of $C$ is $$C = \left\{
{\begin{array}{*{20}{c}}
   {\left( {0,0} \right),{\rm{   \ \ \ \       }}} & {\left( { - 1 + i,1} \right),{\rm{  }}} & {\left( {1 - i, - 1} \right),{\rm{  }}} & {\left( { - 1 - i,i} \right),}  \\
   {\left( {1 + i, - i} \right),{\rm{    }}} & {\left( { - 1 - 2i,2} \right),} & {\left( {1 + 2i, - 2} \right),} & {\left( {2 - i,2i} \right),}  \\
   {\left( { - 2 + i, - 2i} \right),} & {\left( { - 2,1 + i} \right),{\rm{ }}} & {\left( {2i,1 - i} \right),{\rm{  }}} & {\left( {2, - 1 + i} \right),}  \\
   {\left( {2, - 1 - i} \right),{\rm{   }}} & {\left( { - i,2 - i} \right),{\rm{ }}} & {\left( {i, - 2 + i} \right),{\rm{ }}} & {\left( {1,1 + 2i} \right),{\rm{ }}}  \\
   {\left( { - 1, - 1 - 2i} \right)} & {} & {} & {}  \\
\end{array}} \right\}.$$
The minimum Mannheim distance of the code $C$ is 3 and the minimum
Hamming distance of the code $C$ is 2. Let us assume that at the
receiving end we get the vector $r=(\ -1+i, \ 0\ )$. The minimum
Mannheim distance between $r$ and the codewords of $C$ is 1,
namely, $d_M(r,(\ -1+i,\ 1\ ))=1$. Thus, we can correct this error
with respect to the Mannheim metric. But, we can not correct this
error with respect to the Hamming metric since $d_H(r,(\ -1+i,\ 1\
))=1$ and $d_H(r,(\ 0,\ 0\ ))=1$.
\end{example}

 We now define a block code $C$ of length $n$ over ${G_\pi }$ as a set of
 codewords $c = \left( {\begin{array}{*{20}{c}}
   {{c_0},} & {{c_1},} &  \cdots  & {,{c_{n - 1}}}  \\
\end{array}} \right)$ with coefficients ${c_i} \in {G_\pi }$. Let
${\alpha _1},{\alpha _2} \in {G_\pi }$ be two different elements
of orders $p-1$ such that $\alpha _1^{{{p - 1} \mathord{\left/
 {\vphantom {{p - 1} 4}} \right.
 \kern-\nulldelimiterspace} 4}} = i$ and $\alpha _2^{{{p - 1} \mathord{\left/
 {\vphantom {{p - 1} 4}} \right.
 \kern-\nulldelimiterspace} 4}} =  - i$. Hence, ${x^{{{p - 1} \mathord{\left/
 {\vphantom {{p - 1} 4}} \right.
 \kern-\nulldelimiterspace} 4}}} - i$ and ${x^{{{p - 1} \mathord{\left/
 {\vphantom {{p - 1} 4}} \right.
 \kern-\nulldelimiterspace} 4}}} + i$ are factored as $\left( {x - {\alpha _1}} \right)\left( {x - \alpha _1^5} \right) \cdots \left( {x - \alpha _1^{p - 4}}
 \right)$ and $\left( {x - {\alpha _2}} \right)\left( {x - \alpha _2^5} \right) \cdots \left( {x - \alpha _2^{p - 4}}
 \right)$, respectively. Also, the polynomials ${x^{{{p - 1} \mathord{\left/
 {\vphantom {{p - 1} 2}} \right.
 \kern-\nulldelimiterspace} 2}}} + 1$ and ${x^{p - 1}} - 1$ are factored as
 \begin{equation} \label{eq:2} \left( {x - {\alpha _1}} \right)\left( {x - \alpha _1^5} \right) \cdots \left( {x - \alpha _1^{p - 4}} \right)\left( {x - {\alpha _2}} \right)\left( {x - \alpha _2^5} \right) \cdots \left( {x - \alpha _2^{p - 4}}
 \right)\end{equation} and \begin{equation} \label{eq:3}\left( {{x^{{{p - 1} \mathord{\left/
 {\vphantom {{p - 1} 2}} \right.
 \kern-\nulldelimiterspace} 2}}} + 1} \right)\left( {{x^{{{p - 1} \mathord{\left/
 {\vphantom {{p - 1} 2}} \right.
 \kern-\nulldelimiterspace} 2}}} - 1} \right),\end{equation} respectively. A monic polynomial
 $g(x)$ in ${G_\pi }\left[ x \right]$ is the generator polynomial for a cyclic code if and only if
 $\left. {g(x)} \right|{x^n} \pm 1$, where ${G_\pi }\left[ x
 \right]$ is the set of all polynomials with coefficients in ${G_\pi
 }$. Hence, Using (2), we always can construct two classical codes $C_1$, $C_2$ of length $n=(p-1)/2$ over $G_\pi$ such that ${C_2} \subset
 {C_1}$. \\

 \section{Nonbinary quantum CSS codes}

 Let $p$ be an odd prime, let $p = \pi \overline \pi   \equiv 1\;\left( {\bmod \,4} \right)$. A $p$-ary quantum code $Q$ of length $n$ and size $K$ is a $K-$dimensional subspace of
 a $p^n-$dimensional Hilbert space. This Hilbert space is identified with the $n-$fold tensor product of $p-$dimensional Hilbert space, that
 is, ${\left( {{\mathcal{C}^p}} \right)^{ \otimes n}} = {\mathcal{C}^p} \otimes {\mathcal{C}^p} \cdots
 {\mathcal{C}^p}$, where $\mathcal{C}$ denotes complex numbers. We denote by $\left| u \right\rangle$ the vectors of a distinguished orthonormal basis
 of $\mathcal{C}^p$, where the labels $u$ range over the elements of the finite
 field $H_\pi$. For $u = \left( {\begin{array}{*{20}{c}}
   {{u_0},} & {{u_1},} &  \cdots  & {,{u_{n - 1}}}  \\
\end{array}} \right),v = \left( {\begin{array}{*{20}{c}}
   {{v_0},} & {{v_1},} &  \cdots  & {,{v_{n - 1}}}  \\
\end{array}} \right) \in G_\pi ^n$, let $u \cdot v = \sum
{{u_i}{v_i}}$ be the usual inner product on $G_\pi ^n$. For
$\left( {\left. u \right|v} \right),\left( {\left. {u'} \right|v'}
\right) \in G_\pi ^{2n}$, set $\left( {\left. u \right|v} \right)
* \left( {\left. {u'} \right|v'} \right) = Tr(vu' - v'u)$, where
$Tr:{G_{{\pi ^k}}} \to {G_\pi }$ is the trace map. For the integer
$k=1$ then $\left( {\left. u \right|v} \right)
* \left( {\left. {u'} \right|v'} \right) =(vu' - v'u)$. Let $w=\left( {\left. u \right|v} \right)
- \left( {\left. {u'} \right|v'}
\right)=(u_i-u^{'}_i|v_i-v^{'}_i)=(w_i|w^{'}_i) \ (mod\ \pi)$, for
$ i=0,1,2,...,n-1$, we define the Mannheim weight of $w$ and the
Mannheim distance between $\left( {\left. u \right|v} \right) $
and $ \left( {\left. {u'} \right|v'} \right)$ as

$$w{t_M}\left( w \right) = \left\lceil {{{\left[ \begin{array}{l}
 \left| {{\mathop{\rm Re}\nolimits} \left( {{w_0}} \right)} \right| + \left| {{\mathop{\rm Im}\nolimits} \left( {{w_0}} \right)} \right| +  \cdots  + \left| {{\mathop{\rm Re}\nolimits} \left( {{w_{n - 1}}} \right)} \right| + \left| {{\mathop{\rm Im}\nolimits} \left( {{w_{n - 1}}} \right)} \right| \\
  + \left| {{\mathop{\rm Re}\nolimits} \left( {w_0^{'}} \right)} \right| + \left| {{\mathop{\rm Im}\nolimits} \left( {w_0^{'}} \right)} \right| +  \cdots  + \left| {{\mathop{\rm Re}\nolimits} \left( {w_{n - 1}^{'}} \right)} \right| + \left| {{\mathop{\rm Im}\nolimits} \left( {w_{n - 1}^{'}} \right)} \right| \\
 \end{array} \right]} \mathord{\left/
 {\vphantom {{\left[ \begin{array}{l}
 \left| {{\mathop{\rm Re}\nolimits} \left( {{w_0}} \right)} \right| + \left| {{\mathop{\rm Im}\nolimits} \left( {{w_0}} \right)} \right| +  \cdots  + \left| {{\mathop{\rm Re}\nolimits} \left( {{w_{n - 1}}} \right)} \right| + \left| {{\mathop{\rm Im}\nolimits} \left( {{w_{n - 1}}} \right)} \right| \\
  + \left| {{\mathop{\rm Re}\nolimits} \left( {w_0^{'}} \right)} \right| + \left| {{\mathop{\rm Im}\nolimits} \left( {w_0^{'}} \right)} \right| +  \cdots  + \left| {{\mathop{\rm Re}\nolimits} \left( {w_{n - 1}^{'}} \right)} \right| + \left| {{\mathop{\rm Im}\nolimits} \left( {w_{n - 1}^{'}} \right)} \right| \\
 \end{array} \right]} 2}} \right.
 \kern-\nulldelimiterspace} 2}} \right\rceil$$

 $d_M(\left( {\left. u
\right|v} \right) , \left( {\left. {u'} \right|v'}
\right))=wt_M(w)$, respectively. Let $C \subset G_\pi ^{2n}$. Then
the dual code ${C^ {\bot_*}}$ of $C$ is defined to be $${C^
{\bot_*} } = \left\{ {\left( {u\left| v \right.} \right) \in G_\pi
^{2n}:\left( {\left. u \right|v} \right)*\left( {\left. {u'}
\right|v'} \right) = 0\;{\rm{for}}\;{\rm{all }}\left( {\left. {u'}
\right|v'} \right) \in C} \right\}.$$

\begin{definition}
The unitary operators were defined in \cite{8} as ${X_a}\left| u
\right\rangle =
 \left| { \left( {a + u} \right)} \right\rangle ,\;{Z_b}\left| u \right\rangle  = {\xi ^{
 \left( {bu} \right)}}\left| u \right\rangle \;$, where $a,b$ are elements of the finite fields $F_p$, and $\xi$ is a primitive $p$th root of
 unity.
\end{definition}

\begin{definition}
We define the unitary operators as ${X_a}\left| u \right\rangle =
 \left| {\mu \left( {a + u} \right)} \right\rangle ,\;{Z_b}\left| u \right\rangle  = {\xi ^{{\mu ^{ - 1}}
 \left( {bu} \right)}}\left| u \right\rangle \;$, where $a,b\in G_\pi$, $\xi$ is a primitive $p$th root of unity, and
the function $\mu :{F_p} \to {G_\pi }$ defines  $\mu \left( g
\right) = g - \left[ {{{g\overline \pi  } \mathord{\left/
 {\vphantom {{g\overline \pi  } p}} \right.
 \kern-\nulldelimiterspace} p}} \right]\pi $.
\end{definition}
Also, we define the Hadamard gate as $${H_{gate}} = \frac{1}{\sqrt
{N(\pi )}}\left( {{a_{s,t}}} \right),\:{a_{s,t}} = {\xi ^{\left(
{s - 1} \right)\left( {t - 1} \right)\ \left( {\bmod \;p}
\right)}},\:\ 1 \le s,t \le p=\pi.\overline \pi=N(\pi).$$
 For example,
let $\pi=2+i$. Then,$${H_{gate}} = \frac{1}{\sqrt {5}}\left(
{\begin{array}{*{20}{c}}
   1 & 1 & 1 & 1 & 1  \\
   1 & \xi  & {{\xi ^2}} & {{\xi ^3}} & {{\xi ^4}}  \\
   1 & {{\xi ^2}} & {{\xi ^4}} & \xi  & {{\xi ^3}}  \\
   1 & {{\xi ^3}} & \xi  & {{\xi ^4}} & {{\xi ^2}}  \\
   1 & {{\xi ^4}} & {{\xi ^3}} & {{\xi ^2}} & \xi   \\
\end{array}} \right).$$ Note that ${H_{gate}}H_{gate}^\dag =H_{gate}^\dag {H_{gate}} =
{I_p}$, where $H_{gate}^\dag$ denotes the conjugate transpose of
$H_{gate}$ and $I_p$ denotes the identity matrix in $p$
dimensions.

\begin{theorem}(CSS Code Construction) Let $C_1$ and $C_2$ denote two classical linear codes over $G_\pi$ with
the parameters ${\left[ {n,{k_1},{d_{M_1}}} \right]_\pi}$ and
${\left[ {n,{k_2},{d_{M_2}}} \right]_\pi}$ such that ${C_2}
\subseteq {C_1}$. Then, there exists an ${\left[ {\left[ {n,{k_1}
- {k_2},d_M} \right]} \right]_\pi}$ quantum code with minimum
distance ${d_M} = \min \left\{ {{d_{{M_1}}},{d_{{M_2}}^\bot}}
\right\}$, where $d_{{M_2}}^\bot$ denotes the minimum Mannheim
distance of the dual code $C_2^\bot$ of the code $C_2$.
 \\
\end{theorem}

\begin{proof} Let $x$ be a codeword of $C_1$. Then, we define the
quantum state $$\left| {x + {C_2}} \right\rangle  =
\frac{1}{{\sqrt {{C_2}} }}\sum\limits_{y \in {C_2}} {\left| {x +
y} \right\rangle } ,$$ where + is bitwise addition modulo $\pi$.
If $x^{'}$ is an element of $C_1$ such that $x-x^{'} \in C_2$
then, $\left| {x + {C_2}} \right\rangle  = \left| {{x^{'}} +
{C_2}} \right\rangle $, and thus the state $\left| {x + {C_2}}
\right\rangle $ depends only upon the coset of $C_1/C_2$. The
number of cosets of $C_2$ in $C_1$ is equal to ${{\left| {{C_1}}
\right|} \mathord{\left/
 {\vphantom {{\left| {{C_1}} \right|} {\left| {{C_2}} \right|}}} \right.
 \kern-\nulldelimiterspace} {\left| {{C_2}} \right|}}$ so the
 dimension of the quantum code is $N(\pi)^{k_1-k_2}$. Hence, we
 define the quantum code $Q_{C_1,C_2}$ as the vector space spanned by
 t he state $\left| {x + {C_2}} \right\rangle $ for all $x \in
 C_1$. Therefore, the quantum code $Q_{C_1,C_2}$ is an ${\left[ {\left[ {n,{k_1} - {k_2},{d_M}} \right]} \right]_\pi
 }$.\

 \

 We now explain the minimum Mannheim distance $d_M$ of the quantum code
 $Q_{C_1,C_2}$ equals $= \min \left\{ {{d_{{M_1}}},{d_{{M_2}}}}
 \right\}$. Suppose that a bit flip error occurs at only one qubit in
 $n$ qubit and a phase flip error occurs at only one qubit in $n$ qubit. If  $\left| {x + {C_2}} \right\rangle $ was the original
 state then the corrupted state is $$\frac{1}{{\sqrt {{C_2}} }}\sum\limits_{y \in {C_2}} {{\xi ^{{\mu ^
 { - 1}}\left( {\left( {x + y} \right){\widehat{e}_2}} \right)}}\left| {\left( {x + y + {\widehat{e}_1}} \right)\;\bmod \,\pi } \right\rangle }
 .$$ To detect where bit flip error occurred it is convenient to
 introduce an ancilla containing sufficient qubits to store the
 syndrome for the code $C_1$, and initially in the all zero state $\left| 0 \right\rangle
 $. We use reversible computation to apply the parity check matrix
 $H_1$ for the code $C_1$, taking $\left| {x + y + {\widehat{e}_1}} \right\rangle \left| 0 \right\rangle
 $ to $\left| {x + y + {\widehat{e}_1}} \right\rangle \left| {{H_1}\left( {x + y + {\widehat{e}_1}} \right)} \right\rangle  = \left| {x + y + {\widehat{e}_1}}
  \right\rangle \left| {{H_1}\left( {{\widehat{e}_1}} \right)} \right\rangle
  $, since $(x+y)\in C_1$ is annihilated by the parity check
  matrix. The effect of this operation is to produce the state: $$\frac{1}{{\sqrt {{C_2}} }}\sum\limits_{y
  \in {C_2}} {{\xi ^{{\mu ^{ - 1}}\left( {\left( {x + y} \right){\widehat{e}_2}} \right)}}\left| {x + y + {\widehat{e}_1}}
  \right\rangle \left| {{H_1}\left( {{\widehat{e}_1}} \right)} \right\rangle }
  .$$Error detection for the bit flip error is completed by
  measuring the ancilla to obtain the result $H_1(\widehat{e}_1)$ and
  discarding the ancilla. This shows that the bit flip error correcting capacity of the
  quantum code $Q_{C_1,C_2}$ depends on the classical code $C_1$.
  We now show the phase flip error correcting capacity of the
  quantum code $Q_{C_1,C_2}$ depends on the dual code $C_2^ \bot $
 of the classical code $C_2$. The latest state of the corrupted
 state, discarding the ancilla, is:
 $$\frac{1}{{\sqrt {{C_2}} }}\sum\limits_{y \in {C_2}} {{\xi ^{{\mu ^{ - 1}}\left( {\left( {x + y} \right){\widehat{e}_2}} \right)}}\left| {x + y} \right\rangle }.
 $$ We apply the Hadamard gates to each qubit, taking the state to
 $$\frac{1}{{\sqrt {{C_2}N{{(\pi )}^n}} }}\sum\limits_z {\sum\limits_{y \in {C_2}} {{\xi ^{{\mu ^{ - 1}}\left( {\left( {x + y} \right)\left( {{\widehat{e}_2} + z} \right)} \right)}}\left| z \right\rangle } }
 ,$$
where the sum is over all possible values for $n$ bit $z$. Setting
$z^{'}\equiv z+\widehat{e}_2 (\bmod \pi)$, we obtain
$$\frac{1}{{\sqrt {{{N{{(\pi )}^n}} \mathord{\left/
 {\vphantom {{N{{(\pi )}^n}} {{C_2}}}} \right.
 \kern-\nulldelimiterspace} {{C_2}}}} }}\sum\limits_{{z^{'}} \in C_2^ \bot } {{\xi ^{{\mu ^{ - 1}}\left( {x{z^{'}}} \right)}}\left| {{z^{'}} + {\widehat{e}_2}} \right\rangle }.
 $$ Note that if ${z^{'}} \in C_2^ \bot $ then $\sum\nolimits_{y \in
{C_2}} {{\xi ^{{\mu ^{ - 1}}(y{z^{'}})}}}  = \left| {{C_2}}
\right|$, and if ${z^{'}} \notin C_2^ \bot $ then
$\sum\nolimits_{y \in {C_2}} {{\xi ^{{\mu ^{ - 1}}(y{z^{'}})}}}  =
0$. This looks just like a bit flip error described by the vector
$\widehat{e}_2$. To determine the error $\widehat{e}_2$, we
introduce an ancilla qubit and reversibly apply the parity check
matrix $H_2$ for $C_2^ \bot $ to obtain $H_2\widehat{e}_2$, and
correct the error $\widehat{e}_2$, obtaining the state
$$\frac{1}{{\sqrt {{{N{{\left( \pi  \right)}^n}} \mathord{\left/
 {\vphantom {{N{{\left( \pi  \right)}^n}} {\left| {{C_2}} \right|}}} \right.
 \kern-\nulldelimiterspace} {\left| {{C_2}} \right|}}} }}\sum\limits_{{z^{'}} \in C_2^ \bot } {{\xi ^{x{z^{'}}}}\left| {{z^{'}}} \right\rangle }
 .$$
The error correcting is completed by applying the inverse Hadamard
gates, $H_{gate}^\dag$, to each qubit. This takes us back to the
initial state with $\widehat{e}_2=0$. Hence, the proof is
completed.
\end{proof}

We use the Mannheim metric to determine the positions
 and the value of the errors $\widehat{e}_1,\widehat{e}_2$.

Let the minimum Mannheim distance of $C_1$ and $C_2^ \bot $ be
$d_m$. Then, the number of the errors corrected by the quantum
code $Q_{C_1,C_2}$ obtained from the classical codes $C_1$, $C_2$
is equal to
$$4\left( {\begin{array}{*{20}{c}}
   n  \\
   1  \\
\end{array}} \right) + 4^2\left( {\begin{array}{*{20}{c}}
   n  \\
   2  \\
\end{array}} \right) +  \cdots  + 4^t\left( {\begin{array}{*{20}{c}}
   n  \\
   t  \\
\end{array}} \right),$$ where $t = \left\lfloor {{{\left( {{d_m} - 1} \right)} \mathord{\left/
 {\vphantom {{\left( {{d_m} - 1} \right)} 2}} \right.
 \kern-\nulldelimiterspace} 2}} \right\rfloor $ and the symbol $\left( {\begin{array}{*{20}{c}}
    \cdot   \\
    \cdot   \\
\end{array}} \right)$
 gives the binomial coefficient.

\begin{theorem} Let $C = \left( {{C_2}\left| {C_1^ \bot } \right.} \right)$
 be code in $G_\pi^{2n}$ such that $C \subset {C^{{ \bot _ *
}}}$, where $C_1$ and $C_2$ denote two classical codes, and $C_1^
\bot$ denotes the dual code of $C_1$. Then, there exists an
${\left[ {\left[ {n,K,{d_M}} \right]} \right]_\pi }$ quantum code
with the minimum distance
$${d_M} = \min \left\{ {w{t_M}(w):\;w \in {C^{{ \bot _ *
}}}\backslash C} \right\},$$ where $K=dim(C^{{ \bot _ *
}})-dim(C)$.

\end{theorem}
The proof of Theorem 2 can be easily seen from the proof of
Theorem 1.

\begin{example} Let $\pi=4+i$. Let the generator polynomial of
the code $C_1$ be $g_1(x)=1+2i+(-1+i)x-ix^2+x^3$ and let the
generator polynomial of the code $C_2$ be
$g_2(x)=1-i+(2-i)x+(-1+i)x^2-ix^3-ix^4+x^5$. Hence, using the
codes $C_1$ and $C_2$ we obtain a quantum code with parameters
$[[8,2,5]]_{4+i}$ with respect to the Mannheim metric since the
minimum distance of $C_1$ and $C_2^\bot$ are 5.  Let the quantum
state $\left| \psi \right\rangle = \left| {\begin{array}{*{20}{c}}
   {1 - i}, & {2 - i}, & { - 1 + i}, & { - i}, & { - i}, & 1, & 0, & 0  \\
\end{array}} \right\rangle $. If the operator $IIIX_1X_1III$ acts
on this state, then the corrupted state becomes $\left|
{\begin{array}{*{20}{c}} {1 - i}, & {2 - i}, & { - 1 + i}, & {1 - i} ,& {1 - i}, & 1, & 0, & 0  \\
\end{array}} \right\rangle$. The quantum code $[[8,2,5]]_{4+i}$ with respect to the Mannheim
metric overcomes this error since the minimum Mannheim distance of
the classical code $C_1$ is equal to 5. Also, the number of the
bit flip errors corrected by this quantum code is $480$.

\end{example}

On the other hand, let $F_{17}$ be a finite field of
characteristic $17$ and let $C_1$, $C_2$ be the classical codes
with respect to the Hamming metric such that ${C_2} \subset
{C_1}$. Then a quantum code with parameters $[[8,2,4]]_{17}$ can
be obtained. The code $[[8,2,4]]_{17}$ is a maximum distance
separable (shortly MDS) since this code attains the quantum
singleton bound, namely, $17^2=17^{8-2.4+2}$. Also, the number of
the bit flip errors corrected by this quantum code is $128$. For
the length $n=8$,  a quantum code having the minimum distance
greater than 4 is not obtained with respect to the Hamming metric.
So, it is obvious that the quantum code obtained here is better
than the quantum code of the same length with respect to the
Hamming metric.\ \

\

In Table I, some CSS codes compared with respect to the Hamming
metric and the Mannheim metric are given. The CSS codes
constructed from classical codes with respect to the Hamming
metric can be found in the HM column of Table I. The CSS codes
constructed from classical codes with respect to the Mannheim
metric can be found in the MM column of Table I. It is obvious
that, some of the quantum codes obtained in this paper can correct
more errors than the quantum MDS codes of the same length given in
Table I. Using a computer program, we compute the minimum Mannheim
distance of the codes given in Table I.\

\bigskip
\bigskip
\bigskip
\bigskip
\bigskip
\bigskip

\begin{center} \small
{\scriptsize {Table I: Some CSS codes compared with respect to the
Hamming metric and the Mannheim metric.}} {\small \centering}
\begin{tabular}{|c|c|c|c|c|c|c|}
  \hline
  $p$ & $\alpha_1$ & $\alpha_2$ & $h_1$ & $g_2$ & HM & MM \\
  \hline
  5  &    $i$ &  $-i$  & $x^3-ix^2-x+i$           & $x^3+ix^2-x-i$ & ${\left[ {\left[ {4,2,2} \right]} \right]_{2 + i}}$ &${\left[ {\left[ {4,2,2} \right]} \right]_{2 + i}}$ \\
  \hline 13 &      2 &     -2 &$\left( {1 - i} \right) - x + {x^2}$& $\left( {1 - i} \right) + x + {x^2}$ & ${\left[ {\left[ {6,2,3} \right]} \right]_{3 + 2i}}$ &${\left[ {\left[ {6,2,4} \right]} \right]_{3 + 2i}}$ \\
  \hline 13 &      2 &     -2 & $- i - 2x + 2i{x^2} + {x^3}$       & $1 - i + x + {x^2}$ & ${\left[ {\left[ {6,1,3} \right]} \right]_{3 + 2i}}$ & ${\left[ {\left[ {6,1,4} \right]} \right]_{3 + 2i}}$ \\
  \hline 13 &      2 &     -2 & $x-2$                               & $x+2$ & ${\left[ {\left[ {6,4,2} \right]} \right]_{3 + 2i}}$ & ${\left[ {\left[ {6,4,2} \right]} \right]_{3 + 2i}}$ \\
  \hline 13 &      2 &     -2 & $- 1 + ix + {x^2}$                 & $- i + \left( { - 1 + i} \right)x + {x^2}$ & ${\left[ {\left[ {6,2,2} \right]} \right]_{3 + 2i}}$ & ${\left[ {\left[ {6,2,2} \right]} \right]_{3 + 2i}}$ \\
\hline 17 &  $1+i$ & $-2+i$ & $\begin{array}{l}
 -1+i+(2-i)x \\
   + (1-i)x^2-i{x^3} \\
   +i{x^4}+x^5 \\
 \end{array}$ & $\begin{array}{l}
 -i +(-2i)x \\
  +x^3+x^4 \\
 \end{array}$ & ${\left[ {\left[ {8,1,4} \right]} \right]_{4 + i}}$ & ${\left[ {\left[ {8,1,5} \right]} \right]_{4 + i}}$ \\
  \hline 17 &  $1+i$ & $-2+i$ & $\begin{array}{l}
 \left( { - 2 + i} \right) + \left( {1 + i} \right)x \\
  + \left( {2 - i} \right){x^2} + {x^3} \\
 \end{array}$ & $\begin{array}{l}
 \left( {2 - i} \right) + \left( {1 + i} \right)x \\
  - \left( {2 - i} \right){x^2} + {x^3} \\
 \end{array}$ & ${\left[ {\left[ {8,2,4} \right]} \right]_{4 + i}}$ & ${\left[ {\left[ {8,2,5} \right]} \right]_{4 + i}}$ \\
  \hline 17 &  $1+i$ & $-2+i$ & $- 1 + \left( {1 + i} \right)x + {x^2}$ & $- 1 - \left( {1 + i} \right)x + {x^2}$ & ${\left[ {\left[ {8,4,3} \right]} \right]_{4 + i}}$ & ${\left[ {\left[ {8,4,\ge3} \right]} \right]_{4 + i}}$ \\
  \hline 17 &  $1+i$ & $-2+i$ & $- \left( {1 + i} \right) + x$ & $1+i+x$ & ${\left[ {\left[ {8,6,2} \right]} \right]_{4 + i}}$ &${\left[ {\left[ {8,6,\ge2} \right]} \right]_{4 + i}}$ \\
 \hline  19 & $-1+i$ & $-2+i$ & $-2+x$ & $2+x$ & ${\left[ {\left[ {14,12,2} \right]} \right]_{5 + 2i}}$ &${\left[ {\left[ {14,12,\ge2} \right]} \right]_{5 + 2i}}$ \\
  \hline
\end{tabular}

\end{center}

\end{document}